\documentclass[11pt,twoside]{amsart}
\usepackage[foot]{amsaddr}

\usepackage{amsmath, amssymb, amsfonts, amsthm}

\usepackage{xspace}
\usepackage{chngpage}
\usepackage[usenames]{color}
\usepackage[latin1]{inputenc}
\usepackage[american]{babel}
\usepackage{graphicx}
\usepackage{epsfig}
\usepackage{color}
\usepackage{multirow}
\usepackage[margin=0pt,font=small,labelfont=bf,indention=.45cm]{caption}  
\usepackage{latexsym}
\usepackage{pstricks}

\usepackage{setspace}

\usepackage[numbers,sort&compress,longnamesfirst,sectionbib]{natbib}

\shortcites{GhoshLMMPRRTZ99}

\newcommand{\cal}{\mathcal}

\newcommand{\whp}{\textit{whp}\xspace}

\newtheorem{thm}{Theorem}  
\newtheorem{lem}[thm]{Lemma}

\newtheorem{rem}{Remark}

\newtheorem{pro}[thm]{Proposition}

\numberwithin{thm}{section}
\numberwithin{rem}{section}

\newcommand{\thmref}[1]{Theorem~\ref{thm:#1}}

\newcommand{\proref}[1]{Proposition~\ref{pro:#1}}

\newcommand{\lemref}[1]{Lemma~\ref{lem:#1}}
\newcommand{\lemrefs}[2]{Lemmas~\ref{lem:#1} and~\ref{lem:#2}}

\newcommand{\figref}[1]{Figure~\ref{fig:#1}}

\newcommand{\secref}[1]{Section~\ref{sec:#1}}

\newcommand{\eq}[1]{equation~\eqref{eq:#1}}

\renewcommand{\Pr}[1]{\ensuremath{\operatorname{\mathbf{Pr}}\left[#1\right]}}
\newcommand{\Pro}[1]{\ensuremath{\operatorname{\mathbf{Pr}}\left[#1\right]}}
\newcommand{\Ex}[1]{\ensuremath{\operatorname{\mathbf{E}}\left[#1\right]}}
\newcommand{\bal}{\mathsf{b}}
\newcommand{\Even}{\mathsf{Even}}
\newcommand{\Odd}{\mathsf{Odd}}
\newcommand{\ccc}{\ensuremath{\mathsf{CCC}}}

\newcommand{\Range}[1]{\ensuremath{\operatorname{\mathbf{Range}}\left[#1\right]}}

\newcommand{\bx}{\ensuremath{\mathbf{x}}}
\newcommand{\by}{\ensuremath{\mathbf{y}}}
\newcommand{\cB}{\ensuremath{\mathcal{B}}}
\newcommand{\one}{\textbf{1}}

\newcommand{\cM}{\ensuremath{\mathcal{M}}}
\newcommand{\bP}{\ensuremath{\mathbf{P}}}
\newcommand{\bA}{\ensuremath{\mathbf{A}}}
\newcommand{\bB}{\ensuremath{\mathbf{B}}}
\newcommand{\bQ}{\ensuremath{\mathbf{Q}}}

\newcommand{\Ml}{\ensuremath{M^{(i)}}}
\newcommand{\MlOdd}{\ensuremath{M^{(i)}_{\Odd}}}
\newcommand{\MlEven}{\ensuremath{M^{(i)}_{\Even}}}

\newcommand{\Mi}{\ensuremath{M^{(i)}}}

\newcommand{\MtOdd}{\ensuremath{M^{(t)}_{\Odd}}}
\newcommand{\MtEven}{\ensuremath{M^{(t)}_{\Even}}}
\newcommand{\Mtt}{\ensuremath{M^{[t_1\!,t_2]}}}
\newcommand{\MttOdd}{\ensuremath{M^{[t_1\!,t_2]}_{\Odd}}}
\newcommand{\myM}{\ensuremath{\mathcal{M}}}

\newcommand{\Oh}{\mathcal{O}}

\newcommand{\tsum}{\textstyle\sum}
\newcommand{\tprod}{\textstyle\prod}

\DeclareSymbolFont{AMSb}{U}{msb}{m}{n}

\newcommand{\Q}{{\mathbb{Q}}}
\newcommand{\Z}{{\mathbb{Z}}}
\newcommand{\R}{{\mathbb{R}}}

\def\mod{\operatorname{mod}}

\let\oldsqrt\sqrt
\def\hksqrt{\mathpalette\DHLhksqrt}
\def\DHLhksqrt#1#2{\setbox0=\hbox{$#1\oldsqrt{#2\,}$}\dimen0=\ht0
   \advance\dimen0-0.2\ht0
   \setbox2=\hbox{\vrule height\ht0 depth -\dimen0}%
   {\box0\lower0.4pt\box2}}
\renewcommand\sqrt\hksqrt
\renewcommand{\leq}{\leqslant}
\renewcommand{\geq}{\geqslant}

\renewcommand{\epsilon}{\varepsilon}

\newcommand{\alphasmoothed}{$\alpha$\nobreakdash-perturbed}

\allowdisplaybreaks[1]

\title[Smoothed Analysis of Balancing Networks]{Smoothed Analysis of Balancing Networks$^*$}
\thanks{$^*$ A conference version~\citep{FriedrichSV09} appeared in the
    36th International Colloquium on Automata, Languages and Programming (ICALP 2009).    
    This work was done while the first two authors
    were postdoctoral fellows at International Computer Science Institute Berkeley supported
    by the German Academic Exchange Service (DAAD)
    and the third author was a postdoctoral fellow
    at the Computer Science Division of the University of California Berkeley.}

\author{Tobias Friedrich$^1$}
\address{$^1$ Max-Planck-Institut f\"ur Informatik, Saarbr\"ucken, Germany}

\author{Thomas Sauerwald$^2$}
\address{$^2$ Simon Fraser University, Burnaby, Canada}

\author{Dan Vilenchik$^3$}
\address{$^3$ 
Department of Mathematics, University of California Los Angeles, CA, USA}

\begin{document}

\begin{abstract}
In a balancing network each processor has an initial collection
of unit-size jobs (tokens) and in each round, pairs of processors connected by balancers
split their load as evenly as possible.
An excess token (if any) is placed according to some predefined rule. As it turns out, this rule crucially affects
the performance of the network. In this work we propose a model that studies this effect. We suggest a model bridging
the uniformly-random assignment rule, and the arbitrary one (in the spirit of smoothed-analysis). We start with an arbitrary assignment of balancer directions and then flip each assignment with probability $\alpha$ independently.
For a large class of balancing networks our result implies that
after $\Oh(\log n)$ rounds the discrepancy is $\Oh( (1/2-\alpha) \log n + \log \log n)$ with high probability. This matches and generalizes known upper bounds for $\alpha=0$ and $\alpha=1/2$. We also show that a natural network matches the upper bound for any $\alpha$.
\end{abstract}

\maketitle


\section{Introduction}

\begin{figure}[t]
    \centering
    \begin{picture}(0,0)%
    \includegraphics{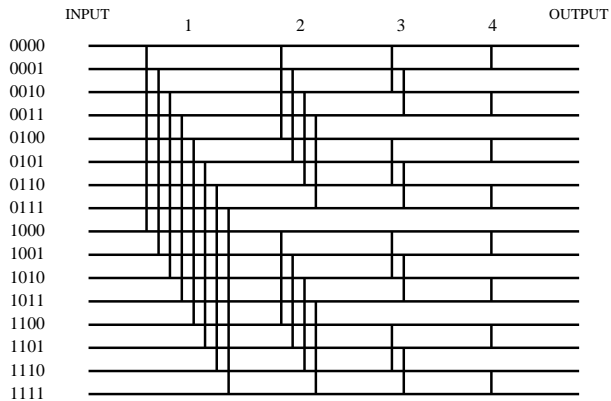}%
    \end{picture}%
    \setlength{\unitlength}{4144sp}%
    \begingroup\makeatletter\ifx\SetFigFont\undefined%
    \gdef\SetFigFont#1#2#3#4#5{%
      \reset@font\fontsize{#1}{#2pt}%
      \fontfamily{#3}\fontseries{#4}\fontshape{#5}%
      \selectfont}%
    \fi\endgroup%
    \begin{picture}(3676,2368)(0,-1483)
    \end{picture}%
    \caption{The network $\ccc_{16}$.}
    \label{fig:ccc16}
\end{figure}
In this work we are concerned with two topics whose name contains the word
``smooth'', but in totally different meaning. The first is \emph{balancing (smoothing)
networks}, the second is \emph{smoothed analysis}. Let us start by introducing
these two topics, and then introduce our contribution -- interrelating the two.

\subsection{Balancing (smoothing) networks}

In the standard abstraction of \emph{smoothing} (balancing) networks~\citep{AHS94}, processors are modeled as the
vertices of a graph and connection between them as edges. Each process has an initial collection
of unit-size jobs (which we call tokens). Tokens are routed through the network by transmitting
tokens along the edges according to some local rule. We measure the quality of such a balancing procedure by the maximum difference between the number of tokens at any two vertices at the end.

The local scheme of communication we study is a \emph{balancer} gate: the number
of tokens is split as evenly as possible between the communicating vertices with
the excess token (if such remains) routed to the vertex towards which the
balancer points. More formally, the balancing network consists of $n$ vertices
$v_1,v_2,\ldots, v_n$, and $m$ matchings (either perfect or not)
$M_1,M_2,\ldots,M_m$. We associate with every matching edge a balancer gate
(that is, we think of the edges as directed edges). At the beginning of the first
iteration, $x_j$ tokens are placed in vertex $v_j$, and at every iteration
$r=1,\ldots,m$, the vertices of the network perform a balancing operation
according to the matching $M_r$ (that is, vertices $v_i$ and $v_j$ interact if
$(v_i,v_j) \in M_r$).

One motivation for considering smoothing networks comes from the
server-client world. Each token represents a client request for some service;
the service is provided by the servers residing at the vertices. Routing tokens through the network must
ensure that all servers receive approximately the same number of tokens, no
matter how unbalanced the initial number of tokens is (cf.~\citep{AHS94}).
More generally, smoothing networks are attractive for
multiprocessor coordination and load balancing applications where low-contention
is a requirement; these include \textit{producers-consumers} \citep{HS08} and
distributed numerical computations \citep{BT97}. Together with \emph{counting
networks}, smoothing networks have been studied quite extensively since
introduced in the seminal paper of \citet{AHS94}.

\medskip

\citet{HT06,HT06b} initiated the study of the
{\it \ccc\ network} (cube-connected-cycles, see \figref{ccc16}) as a smoothing
network\footnote{Actually, they considered the so-called {\it block network}. However, it was observed in \cite{MS08} that the block network is isomorphic to the \ccc-network and therefore we will stick to the latter in the following.}. For the special case of the $\ccc$, sticking to previous conventions, we adopt a ``topographical''
view of the network, thus calling the vertices \emph{wires}, and looking at the left-most side of the network as the ``input'' and the right-most as the ``output''.
In the $\ccc$, two wires at layer $\ell$ are connected by a balancer if the respective
bit strings of the wires differ exactly in bit $\ell$. The \ccc\ is
a canonical network in the sense that it has the smallest possible depth
(number of rounds)
of $\log n$ as a smaller depth cannot ensure any discrepancy independent of the initial one.
Moreover, it has been used in more advanced constructions
such as the {\it periodic (counting) network}~\citep{AHS94,DPRS89}.

As it turns out, the initial setting of the balancers' directions is crucial.
Two popular options are an arbitrary orientation or a uniformly random one.
A maximal discrepancy of $\log n$ was established for the $\ccc_n$ for an arbitrary initial orientation~\citep{HT06b}.
For a random initial orientation of the $\ccc_n$, \citep{HT06} show a discrepancy of at most $2.36 \sqrt{\log n}$ for the $\ccc_n$
(this holds \whp\footnote{Writing $\whp$ we mean
with probability tending to $1$ as $n$ goes to infinity.} over the random initialization), which was improved
by \citet{MS08} to $\log \log n + \Oh(1)$ (and a matching lower bound).

Results for more general networks have been derived in \citet{RSW98} for
arbitrary orientations. For expander graphs, they show an $\Oh(\log n)$-discrepancy
after $\Oh(\log (Kn))$ rounds when $K$ is the discrepancy of the initial load vector. This was recently strengthened assuming
the orientations are set randomly and in addition the matchings themselves are
chosen randomly \citep{FS09}. Specifically, for expander graphs one can achieve within the same number of rounds a discrepancy of only $\Oh(\log \log n)$.

\subsection{Smoothed analysis}

Let us now turn to the second meaning of ``smoothed''. Smoothed analysis comes to
bridge between the random instance, which typically has a very specific
``unrealistic'' structure, and the completely arbitrary instance, which in many
cases reflects just the worst case scenario, and is thus over-pessimistic in
general. In the smoothed analysis paradigm, first an adversary generates an
input instance, then this instance is randomly perturbed.

The smoothed analysis paradigm was introduced by \citeauthor{SpielmanT01} in 2001
\citep{SpielmanT01} to help explain why the simplex algorithm for linear
programming works well in practice but not in (worst-case) theory.  They
considered instances formed by taking an arbitrary constraint matrix and
perturbing it by adding independent Gaussian noise with variance $\varepsilon$
to each entry.  They showed that, in this case, the shadow-vertex pivot rule
succeeds in expected polynomial time. Independently, \citet{BohmanFM03}
studied the issue of Hamiltonicity in a dense graph when
random edges are added. In the context of graph optimization problems we can also
mention \cite{KrivSudakovTetail,DiamOfRandGraph}, in the context of $k$-SAT
\cite{FeigeSmoothedRefute,Walksat}, and in various other problems
\cite{smoothedIP,smoothedTrees,smoothedPolytope,smoothedKMeans}.

Our work joins this long series of papers studying perturbed instances in a
variety of problems. Specifically in our setting we study the following
question: what if the balancers were not set completely adversarially but also
not in a completely random fashion. Besides the mathematical and analytical
challenge that such a problem poses, in real network applications one may not
always assume that the random source is unbiased, or in some cases one will not
be able to quantitatively measure the amount of randomness involved in the
network generation. Still it is desirable to have an estimate of the typical
behavior of the network. Although we do not claim that our smoothed-analysis
model captures all possible behaviors, it does give a rigorous and tight
characterization of the tradeoff between the quality of load balancing and the
randomness involved in setting the balancers' directions, under rather natural
probabilistic assumptions.

As far as we know, no smoothed analysis framework was suggested to a networking
related problem.
Formally, we suggest the following framework.


\subsection{The Model}\label{sec:TheModel}

We define both the smoothed-analysis aspect of the model, and the load-balancing
one. For the load balancing part, our model is similar to (and, as we will shortly explain, a generalization of) the periodic balancing
circuits studied in \citep{RSW98}. We think of the balancing network in terms of an $n$-vertex graph.
The processors in the network are the vertices of the graph,
and balancers connecting processors are the (directed) edges of the graph.

Before we proceed, since what follows is
somewhat heavy on notation and indices, it will be helpful for the reader to
bear in mind the following legend: we use superscripts (in round brackets) to
denote a time stamp, and subscripts to denote an index. In subscripts, we use the vertices
of the graph as indices (thus assuming some ordering of the vertex set).
For example, $\mathbf{A}^{(i)}_{u,v}$ stands for the $(u,v)$-entry in matrix $\mathbf{A}^{(i)}$, which corresponds to time/round $i$.

Let $M^{(1)},\ldots,M^{(T)}$ be an arbitrary sequence of $T$
(not necessarily perfect) matchings. With each matching $\Mi$ we associate a matrix $\bP^{(i)}$ with
$\bP^{(i)}_{uv}=1/2$ if $u$ and $v$ are matched in $\Mi$,
$\bP^{(i)}_{uu}=1/2$ if $u$ is matched in $\Mi$,
$\bP^{(i)}_{uu}=1$ if $u$ is not matched in $\Mi$,
and $\bP^{(i)}_{uv}=0$ otherwise.

In round $i$, every two vertices matched in $\Mi$ perform a balancing operation.
That is, the sum of the number of tokens in both vertices is split evenly between
the two, with the remaining token (if exists) directed to the vertex pointed by
the matching edge.

\begin{rem}
    In \textbf{periodic balancing networks} (see \citep{RSW98} for example) an ordered
    set of $d$ (usually perfect) matchings is fixed. Every round of balancing is a
    successive application of the $d$ matchings. Our model is a (slight) generalization of the latter.
\end{rem}

Let us now turn to the smoothed-analysis part. Given a balancing network consisting of a set $T$
of directed matchings, an\textbf{ $\alpha$\nobreakdash-perturbation} of the network is a flip of direction for every edge
with probability $\alpha$ independently of all other edges. Setting $\alpha=0$ gives the completely ``adversarial model'', and $\alpha=1/2$ is the uniform random case.

\begin{rem} For our results, it suffices to consider $\alpha \in [0,1/2]$. The case $\alpha > 1/2$ can be reduced
to the case $\alpha \leq 1/2$ by flipping the initial orientation of all balancers and taking $1-\alpha$ instead of $\alpha$.
It is easy to see that both distributions are identical.
\end{rem}


\subsection{Our Contribution}
For a load vector $\bx$, its discrepancy is defined to be $\max_{u,v} |\bx_u-\bx_v|$. 
We use $e_u$ to denote the unit vector whose all entries are 0 except the $u^{th}$. For a matrix $A$, $\lambda(A)$ stands for the second largest eigenvalue of $A$ (in absolute value).
Unless stated otherwise, $\|z\|$ stands for the $\ell_2$-norm of the vector $z$. In the following, we will assume an ordering of the vertices from $1$ to $n$. When we write $(u,v) 
\in E$, we refer to the case where $u$ and $v$ are connected by an undirected edge and $u < v$.

\begin{thm} \label{thm:upper}
Let $G$ be a balancing network with matchings $M^{(1)},\ldots,M^{(T)}$.
For any two time stamps $t_1,t_2$ satisfying $t_1 < t_2\leq T$,
and any input vector with initial discrepancy $K$, the discrepancy at time step $t_2$ in
\alphasmoothed\ $G$ is \whp at most
   \begin{align*}
   	  (t_2-t_1)
   	  + 3 \left(\tfrac{1}{2} - \alpha \right) t_1 +\Lambda_1 + \Lambda_2,
   \end{align*}
   where
   \begin{align*}
        \Lambda_1 &= \max_{w \in V} \sqrt{\log n \, \tsum_{i=1}^{t_1}
           \tsum_{[u:v]\in \Ml}
           \left((e_u-e_v)^{T} \left(\tprod_{j=i+1}^{t_2}
           \bP^{(i)}\right) e_w \right)^2},    \\
        \Lambda_2 &=  \lambda \left(\tprod_{i=1}^{t_2} \bP^{(i)}\right) \sqrt{n} K.
   \end{align*}
\end{thm}
Before we proceed let us motivate the result stated in \thmref{upper}. There are
two factors that affect the discrepancy: the fact that tokens are indivisible
(and therefore the balancing operation may not be ``perfect''), and how many balancing
rounds are there. On the one hand, the more rounds there are the more balancing
operations are carried, and the smoother the output is. On the other hand, the
longer the process runs, its susceptibility to rounding errors and arbitrary
placement of excess tokens increases. This is however only a seemingly tension,
as indeed the more rounds there are, the smoother the output is. Nevertheless,
in the analysis (at least as we carry it), this tension plays part.
Specifically, optimizing over these two contesting tendencies is reflected in
the choice of $t_1$ and $t_2$. $\Lambda_2$ is the contribution resulting from
the number of balancing rounds being bounded, and $\Lambda_1$, along with the
first two terms, account for the indivisibly of the tokens. In the cases that
will interest us, $t_1,t_2$ will be chosen so that $\Lambda_1,\Lambda_2$ will be
low-order terms compared to the first two terms.

\medskip

\noindent Our \thmref{upper} also implies the following results:
\begin{itemize}
  \item For the aforementioned periodic setting Theorem \ref{thm:upper} implies the following:
  after $\Oh \left(\log (Kn)/\nu \right)$ rounds
  ($\nu=(1- \lambda(\bQ))^{-1})$, $\bQ$ is the so-called round matrix which corresponds to one period,
  $K$ the initial discrepancy) the discrepancy is $\whp$ at most
    \[
        \Oh \left(\frac{d \, \log (Kn)}{\nu}\cdot\left(\frac{1}{2}-\alpha \right) + \frac{d \log\log n}{\nu}  \right).
    \]
  Setting $\alpha=0$ (and assuming $K$ is polynomial in $n$) we get the result of \citep{RSW98}\footnote{We point out that in the original statement in \cite[Corollary 5]{RSW98}, only the number of periods is counted. Hence, in their statement the number of rounds is by a factor of $d$ smaller.}, and for
  $\alpha=1/2$ we get the result of \citep{FS09}.
  (The restriction on $K$ being polynomial can be lifted but at the price of more cumbersome expressions in Theorem \ref{thm:upper}.
  Arguably, the interesting cases are anyway when the total number of tokens, and in particular~$K$, is polynomial). Complete details in \secref{DerivationOfPeriodicCase}.
  \item For the $\ccc_n$, after $\log n$ rounds the discrepancy is \whp at most
     \[ 3 \left(\tfrac{1}{2} - \alpha \right)\log n+ \log \log n + \Oh(1).\]
     Full details in \secref{DerivationOfCCC}.
\end{itemize}
Let us now turn to the lower bound. Here we consider the all-up-orientation of the balancers of a $\ccc_n$ meaning that before the $\alpha$-perturbation, all balancers are directed to the wire with a smaller number.

\begin{thm}
    \label{thm:lower}
    Consider a $\ccc_n$ with the all-up orientation of the balancers and assume
    that the number of tokens at each wire is uniformly
    distributed over $\{0,1,\ldots,n-1\}$ (independently at each wire).
    The discrepancy of the $\alpha$\nobreakdash-perturbed
    network is $\whp$ at least
    \[
        \max \Big\{\big(\tfrac{1}{2} - \alpha\big)\log n - 2 \log \log n,(1-o(1))\,(\log \log n)/2\Big\}.
    \]
\end{thm}
Two more points
to note regarding the lower bound:
\begin{itemize}
  \item For $\alpha=0$ our lower bound matches
        the experimental findings of \cite{HT06}. The authors examined \ccc's of size
        up to $2^{24}$ where all balancers are set in the same direction and the number of tokens
        at each input is a random number between $1$ and $100,000$.
        Their observation was that the average discrepancy is close to $(\log n)/2$ (which matches our lower bound
        with $\alpha=0$).
\item  The input distribution that we use for the lower bound is arguably more natural
than the tailored and somewhat artificial ones used in previous
lower bound proofs \citep{HT06b,MS08}.
\end{itemize}

Finally, we state a somewhat more technical result that we obtain, which lies in
the heart of the proof of the lower bound and sheds light on the mechanics of
the $\ccc$ in the average case input. In what follows, for a balancer $\bal$, we
let $\Odd(\bal)$ be an indicator function which is 1 if $\bal$ had an excess
token. 
By $\cB_i$ we denote the set of balancers that affect output wire $i$
(that is, there is a path through consecutive layers from the balancers to
the output wire $i$). $\cB_i(\ell)$ is the restriction of $\cB_i$ to balancers in layer $\ell$.

\begin{lem}
    \label{lem:BalancersOddWithProbHalf} Consider a $\ccc_n$ network with any fixed
    orientation of the balancers. Assume that the number of tokens at each wire is uniformly
    distributed over $\{0,1,\ldots,n-1\}$ (independently at each wire). Then every balancer $\bal$ in layer $\ell$, $1 \leq \ell \leq \log n$, satisfies the following properties:
    \begin{itemize}
      \item $\Pr{\Odd(\bal)=1}=1/2$,
      \item moreover, for every wire $i$, $\{ \Odd(\bal) \mid \bal \in {\cal{B}}_i \}$ is a set of independent random variables.
    \end{itemize}
\end{lem}

The proof of the lemma is given in \secref{ProofOfBalancersProp}. Let
us remark at this point that the lemma holding under such weak conditions is rather surprising. First, it is valid regardless
of the given orientation. Secondly, the $\Odd$'s of the balancers that affect the
\emph{same} output wire are independent. While this is obvious for balancers that are in the same layer, it seems somewhat surprising for balancers in subsequent layers that are connected.

\subsection{Paper's Organization}
The remainder of the paper is organized as follows. We set out with the proof of \thmref{lower}~in Section \ref{sec:lower}, preceding the proof of Theorem \ref{thm:upper} in Section \ref{sec:upper}.
The reason is that the lower bound is concerned with the $\ccc_n$ (a special case
of our general model). The techniques used in the proof of the lower bound serve
as a good introduction to the more complicated proof of the upper bound. In Sections \ref{sec:DerivationOfPeriodicCase} and \ref{sec:DerivationOfCCC} we show how to derive the special cases of the periodic balancing network and the $\ccc_n$ from Theorem \ref{thm:upper}. Finally we present experimental results that we obtain for the $\ccc_n$ in Section \ref{sec:conclusions}.

\section{Proof of the Lower Bound}
\label{sec:lower}

As we mentioned before, for the special case of the $\ccc_n$ we adopt a ``topographical''
view of the network: calling the vertices \emph{wires}, the time steps \emph{layers}, the left-most side of the
network the ``input'' and the right-most the ``output''.

The proof outline is
the following. Given an input vector $\bx$ (uniformly distributed over the range
$\{0,\ldots,n-1\}$), we shall calculate the expected divergence from the average
load $\mu=\| \bx \|_1 /n$. The expectation is taken over both the smoothing
operation and the input. After establishing the ``right'' order of divergence (in
expectation) we shall prove a concentration result. One of the main keys to
estimating the expectation is \lemref{BalancersOddWithProbHalf} saying that if the input is
uniformly distributed as above, then for every balancer $\bal$,
$\Pr{\Odd(\bal)=1}=1/2$ (the probability is taken only over the input).

Before proceeding with the proof, let us introduce some further notation.
Let $y_1$ be the number of tokens exiting on the top output wire of the network.
For any balancer $\bal$, $\Psi(\bal)$ is an indicator random variable which takes the value $-1/2$ if the
balancer $\bal$ was perturbed, and $1/2$ otherwise (Recall that all balancers are pointing up before the perturbation takes place).

Using the ``standard'' backward (recursive) unfolding (see also \citep{HT06,MS08}
for a concrete derivation for the $\ccc_n$) we obtain that,
\[
    y_1 =
    \mu + \sum_{\ell=1}^{\log n} 2^{-\log n+\ell}
    \sum_{\bal \in \cB_1(\ell) } \Odd(\bal) \cdot \Psi(\bal).
\]
The latter already implies that the discrepancy of the entire network is at least
\[
    y_1 - \mu
    = \sum_{\ell=1}^{\log n} 2^{-\log n+\ell} \sum_{\bal \in \cB_1(\ell)} \Odd(\bal) \cdot \Psi(\bal),
\]
because there is at least one wire whose output has at most $\mu$ tokens (a further improvement of a factor of $2$ will be obtained by considering additionally the bottom output wire and proving that on this wire only a small number of tokens exit).

\noindent Write $ y_1 - \mu = \sum_{\ell=1}^{\log n} S_{\ell}$, defining for each layer $1 \leq \ell \leq \log n$,
\begin{equation}
    S_{\ell} :=
    2^{-\log n +\ell}\sum_{\bal \in \cB_1(\ell)}  \Odd(\bal)
    \cdot \Psi(\bal).
    \label{eq:Sell}
\end{equation}


\subsection{Proof of $(\tfrac{1}{2}-\alpha)\log n - 2\log \log n$ }
\label{sec:ProofOfpro:logn_LowerBound}
We now turn to bounding the expected value of $S_\ell$.
Using the following facts: (a) the $\Odd(\bal)$ and $\Psi(\bal)$ are independent
(b) \lemref{BalancersOddWithProbHalf} which gives $\Ex{\Odd(\bal)}=1/2$, (c) the simple
fact that $\Ex{\Psi(\bal)}=\tfrac{1}{2}-\alpha$ and (d) the fact that in layer $\ell$ there are
$2^{\log n - \ell}$ balancers which affect output wire 1 (this is simply by the
structure of the $\ccc_n$), we get
\begin{align*}
    \Ex{S_{\ell}}
    &=2^{-\log n +\ell}\sum_{\bal \in \cB_1(\ell)} \tfrac{1}{2} \cdot (\tfrac{1}{2} - \alpha)\\
    &=2^{-\log n +\ell}\cdot 2^{\log n - \ell} \cdot \tfrac{1}{2} \cdot \left(\tfrac{1}{2}-\alpha\right)\\
    &=  \tfrac{1}{2} \, \left(\tfrac{1}{2}-\alpha\right).
\end{align*}
This in turn gives that
\[
    \Ex{y_1 - \mu}
    =\Ex{\sum_{\ell=1}^{\log n} S_{\ell}}
    =\tfrac{1}{2} \left(\tfrac{1}{2}-\alpha\right)\log n.
\]
Our next goal is to claim that typically the discrepancy behaves like the
expectation; in other words, a concentration result.
Specifically, we apply Hoeffdings bound to each layer $S_\ell$ separately. It is
applicable as the random variables $2^{-\log n +\ell} \cdot \Odd(\bal) \cdot
\Psi(\bal)$ are independent for balancers within the same layer (such balancers
concern disjoint sets of input wires, and the input to the network was chosen independently for
each wire). For the bound to be useful we need the range of values for the
random variables to be small. Thus, in the probabilistic argument, we shall be
concerned only with the first $\log n - \log \log n$ layers (the last $\log \log
n$ layers we shall bound deterministically). We use the following Hoeffding bound:
\begin{lem}[Hoeffdings Bound]
    \label{lem:hoeffding}
    Let $Z_1,Z_2,\ldots,Z_n$ be a sequence of independent random variables with
    $Z_i \in [a_i,b_i]$ for each $i$. Then for any number $\varepsilon \geq 0$,
    \[
      \Pro{ \left| \tsum_{i=1}^{n} Z_i - \Ex{ \tsum_{i=1}^{n} Z_i } \right| \geq \varepsilon}
        \leq 2 \cdot \exp \left(-\frac{2 \varepsilon^2}{\sum_{i=1}^n (b_i - a_i)^2} \right).
    \]
\end{lem}
\noindent
For any random variable $X$, let $\Range{X}$ be the difference between the maximum and minimum value that $X$ can attain. For a balancer $\bal \in \cB_1(\ell)$ we plug in,
\begin{align*}
    Z_\bal &= 2^{-\log n +\ell} \cdot \Odd(\bal) \cdot \Psi(\bal),\\
    \varepsilon &= 2^{(\ell-\log n+\log \log n)/2},\\
    \Range{Z_{\bal}}^2&=\left(2^{\ell-\log n}\right)^2,
\end{align*}
and the sum is over $2^{\log n -
\ell}$ balancers in layer $\ell$. Therefore,
\begin{align*}
    &\Pro{ \left| S_{\ell} - \Ex{S_{\ell}} \right|
    \geq 2^{(\ell-\log n+\log \log n)/2} }\\
    &\qquad\leq 2 \exp \left( -\frac{ 2 \cdot 2^{\ell - \log n + \log \log n } }{ 2^{\ell- \log n}  } \right)\\
    &\qquad\leq n^{-1}.
\end{align*}
In turn, with probability at least $1-\log n /n$ (take the union bound over at
most $\log n$ $S_\ell$ terms):
\[
    \sum_{\ell = 1}^{\log n - \log \log n}S_\ell
    \geq \tfrac{1}{2} \, \left(\tfrac{1}{2}-\alpha\right)(\log n - \log \log n)-
          \sum_{\ell = 1}^{\log n - \log \log n}2^{(\ell-\log n+\log \log n)/2}.
\]
The second term is just a geometric series with quotient $\sqrt{2}$, and therefore can be bounded by $\tfrac{1}{1-1/ \sqrt{2}} < 4$.

For the last $\log \log n$ layers, we have that for every $\ell$,
$|S_\ell|$ cannot exceed $\tfrac{1}{2}$, and therefore their contribution, in absolute value is
at most $\tfrac{1}{2} \log \log n$. Wrapping it up, \whp
\begin{align*}
    y_1 - \mu
    =& \sum_{\ell=1}^{\log n} S_{\ell}
    \geq \tfrac{1}{2} \, \left(\tfrac{1}{2}-\alpha\right)(\log n - \log \log n) - 4 - \tfrac{1}{2} \log \log n .
\end{align*}
The same calculation implies that the number of tokens at the \emph{bottom}-most output wire deviates
from $\mu$ in the same way (just in the opposite direction), so 
\begin{align*}
    y_n - \mu
    \leq& - \tfrac{1}{2} \, \left(\tfrac{1}{2}-\alpha\right)(\log n - \log \log n) + 4 + \tfrac{1}{2} \log \log n .
\end{align*}
Hence, the discrepancy is \whp lower bounded by (using the union bound over the top and bottom wire, and not
claiming independence)
\begin{align*}
 y_1 - y_n
 &\geq \left(\tfrac{1}{2}-\alpha\right) \log n - 8 - (\tfrac{3}{2} - \alpha) \log \log n \\
 &\geq \left(\tfrac{1}{2}-\alpha\right) \log n-2\log \log n.
\end{align*}


\subsection{Proof of $(1-o(1))\log \log n /2$}
\label{sec:ProofOfpro:loglogn_LowerBound}

The proof here goes along similar lines to \secref{ProofOfpro:logn_LowerBound}, only
that now we choose the set of balancers we
apply it to more carefully. By the structure of the $\ccc_n$, the last $x$
layers form the parallel cascade of  $n/2^x$ independent \ccc\ subnetworks each
of which has $2^{x}$ wires (by independent we mean that the sets of balancers are
disjoint).

We call a subnetwork \emph{good} if after an $\alpha$\nobreakdash-perturbation
of the all-up initial orientation, all the balancers affecting the top (or bottom) output wire were not flipped (that is, still point up).

The first observation that we make is that $\whp$ (for a suitable choice of $x$, to
be determined shortly) at least one subnetwork is good. Let us prove this fact.

The number of balancers affecting the top (or bottom) wire in one of the subnetworks is
$\sum_{\ell=1}^{x} 2^{\ell-1} \leq 2^{x}.$ In total, there are no more than $2^{x+1}$ affecting both wires. The probability that none of these balancers was flipped
is (using our assumption $\alpha \leq 1/2$) $(1-\alpha)^{2^{x+1}} \geq 2^{-2^{x+1}}.$
Choosing $x= \log \log n - 2$, this probability is at least $n^{-1/2}$;
there are at least $n/\log n$ such subnetworks, thus the probability that none is good is at most
\[
    \left(1-n^{-1/2}\right)^{n/\log n}=o(1).
\]

\noindent
Fix one good subnetwork and let $\mu'$ be the average load at the input to that subnetwork. Repeating the arguments
from \secref{ProofOfpro:logn_LowerBound} (with $\alpha=0$, $\log n$ re-scaled to $x=\log \log n-2$, and now using the
second item in \lemref{BalancersOddWithProbHalf} which guarantees that the probability of $\Odd(\cdot)=1$ is still $1/2$,
for any orientation of the balancers)
gives that in the top output wire of the subnetwork
there are $\whp$ at least $\mu'+(\log \log n)/4-\Oh(\log \log \log n)$ tokens,
while on the bottom output wire there are $\whp$ at most $\mu'-(\log \log n)/4+\Oh(\log \log \log n)$ tokens.
Using the union bound, the discrepancy is $\whp$ at least their difference, that is, at least $(\log \log n)/2-\Oh(\log \log \log n)$.

\subsection{Proof of \lemref{BalancersOddWithProbHalf}}
\label{sec:ProofOfBalancersProp}

The following observation is the key idea in proving \lemref{BalancersOddWithProbHalf}. Recall that ${\cal{B}}_j$ stands for the set
of balancers that affect wire $j$.
For a balancer $\bal$ in layer $\ell$ let $A(\bal)$ describe an assignment of $\Odd(\bal')$ values for all balancers $\bal'$ in preceding layers that affect $\bal$ (that
is, there is a path from $\bal'$ to $\bal$ through consecutive layers).

\begin{lem}\label{lem:odd} Consider a $\ccc_n$ network with any
    orientation of the balancers. Assume that the number of tokens at each wire is uniformly
    distributed over $\{0,1,\ldots,n-1\}$. Consider a balancer $\bal$ in layer $\ell$ with  $1 \leq \ell \leq \log n$ and let $x_1,x_2$ denote the two input wires that go into $\bal$. Then
    for any assignment $A(\bal)$, $x_i \mod \left(n/2^{\ell-1}\right) \, \mid \, A(\bal)$ is uniformly distributed
    over $\{0,1,\ldots,\left(n/2^{\ell-1}\right)-1\}$, for $i=1,2$.
\end{lem}

The lemma easily implies \lemref{BalancersOddWithProbHalf}: Consider a balancer
$\bal$ in layer $\ell$, with $1 \leq \ell \leq \log n$ with two inputs
$x_1,x_2$. \lemref{odd} implies that $x_i \mod \left(n/2^{\ell-1}\right)$ is
uniformly distributed over $\{0,1,\ldots,n/2^{\ell-1}-1\}$, and in particular is
uniform over that range mod~2. Hence both $x_1,x_2$ are odd/even with
probability 1/2. Further observe that by the structure of the \ccc\ network,
the input wires $x_1,x_2$ depend on disjoint sets of wires and balancers. By the two latter
facts it follows that the sum $x_1+x_2$ is odd (or even) with probability $1/2$,
or in other words $\Pr{\Odd(\bal)=1}=1/2$. The independence part follows from the fact that this is true for every conditioning $A(\bal)$ on balancers from previous layers that affect $\bal$ (by \lemref{odd}) and the fact that balancers in the same layer that affect the same output wire are independent by construction, as those balancers depend on disjoint sets of balancers and disjoint sets of input wires.

\begin{proof}[Proof of \lemref{odd}]
We prove the lemma by induction on $\ell$, the layer of the balancer.
The base case is immediate: the input to a balancer in layer $1$ is just the original input,
which is by definition distributed uniformly over $\{0,1,\ldots,n/2^{1-1}-1\}$
which is simply $\{0,1,\ldots,n-1\}$. Assume the lemma is true for all balancers
in layer $\ell$ and consider a balancer in layer $\ell+1$.
Let $x_1,x_2$ be its
two input wires. By the structure of the \ccc, the value on each wire is
determined by a disjoint set of balancers and preceding wires, therefore we can
treat, w.l.o.g., only $x_1$. Let $A_1(\bal)$ be the part of $A(\bal)$ that affects $x_1$.
($A_1(\bal)$ would be the set of balancers so that there is a path from them to $x_1$. Since
the initial load on the input wires is chosen independently at every wire, and the sets of balancers
affecting $x_1$ and $x_2$ are disjoint, indeed only $A_1(b)$ affects $x_1$, and the same applies for $A_2(b)$ and $x_2$).
Thus our goal is to calculate
\[
    \Pro{x_1 \equiv k  \mod \big(n/2^\ell\big) ~\big|~ A_1(\bal)}.
\]
Let $\bal'$ be the balancer in layer $\ell$ whose
one outlet is $x_1$, and let $x_1',x_2'$ be its two inputs.
Recall that $x_1 = \lfloor(x_1'+x_2')/2\rfloor$, and a possible +1 addition in
case this sum is odd and the balancer points in the direction of $x_1$.
Furthermore it is easy to verify that for every $a$ (assume $a$ is even, if odd
write $a+1$ instead)
\[
    a/2
    \equiv k \mod \big(n/2^\ell\big) \Leftrightarrow a
    \equiv 2k \mod \big(n/2^{\ell-1}\big).
\]
Therefore for the event $x_1 \equiv k  \mod \left(n/2^\ell\right)$ to occur, either
\begin{itemize}
    \item $x_1'+x_2' \equiv 2k \mod \left(n/2^{\ell-1}\right)$, or
    \item $x_1'+x_2' \equiv 2k-1 \mod \left(n/2^{\ell-1}\right)$
    (assuming w.l.o.g.\ that the balancer points towards $x_1$, otherwise the sum equals $2k+1$).
\end{itemize}
Let us consider the first case.
\begin{align*}
    &\Pro{x_1'+x_2' \equiv 2k  \mod \left(n/2^{\ell-1}\right) \mid A_1(b)} \\
    &=\Pro{\bigvee_{i=0}^{n/2^{\ell-1}-1} x_1' \equiv i \mod \left(n/2^{\ell-1}\right)
      \wedge  x_2' \equiv 2k-i \mod \left(n/2^{\ell-1}\right) \mid A_1(b)}.
\end{align*}
Now observe that the values of $x_1'$ and $x_2'$ are determined
independently, as again, by the structure of the \ccc, they involve
disjoint sets of balancers and input wires. Similarly to $A_1(b)$ we can define
$A_1'(b)$ and $A_2'(b)$ which correspond to the parts of $A$ affecting $x'_1$ and $x'_2$.
By the structure of the $\ccc$, $A_1'(b)$ and $A_2'(b)$ depend on a disjoint set of input wires (and balancers).
As the input is chosen independently for every wire, $A_2'(b)$ does not affect $x'_1$ and similarly
$A_1'(b)$ does not affect $x'_2$. Thus the latter reduces to
\begin{align*}
    &\sum_{i=0}^{n/2^{\ell-1}-1} \Pro{x_1'\equiv i \mod \left(n/2^{\ell-1}\right) \mid A_1'(b)}\\
    &\phantom{\sum_{i=0}^{n/2^{\ell-1}-1} }\cdot \Pro{x_2'\equiv 2k-i \mod \left(n/2^{\ell-1}\right) \mid A_2'(b)}.
\end{align*}
By the induction hypothesis (applied to the $x_i'$ at layer $\ell-1$), for every
$i$, each of $\Pro{x_1'\equiv i \mod \left(n/2^{\ell-1}\right) \mid A_1'(b)}$ and
$\Pro{x_2'\equiv 2k-i \mod \left(n/2^{\ell-1}\right) \mid A_2'(b)}$ is uniformly distributed
over the range $\{0,\ldots, n/2^{\ell-1}-1\}$, and therefore in particular the
entire expression does not depend on $k$, or, put differently is the same for
every choice of $k$. The same argument holds for the case
$x_1'+x_2' = 2k-1 \mod \left(n/2^{\ell-1}\right)$.
This completes the proof.
\end{proof}

\section{Proof of the Upper Bound}
\label{sec:upper}

We shall derive our bound by measuring the difference between the number of tokens at any vertex
and the average load (as we did in the proof of the lower bound for the $\ccc_n$).
Specifically we shall bound $\max_i |y^{(t)}_i-\mu|$, $y^{(t)}_i$ being the number
of tokens at vertex $i$ at time $t$ (we use $\by^{(t)}=(y_i^{(t)})_{i\in V}$ for the vector of loads at time $t$).
There are two contributions to the divergence from $\mu$ (which we analyze separately):
\begin{itemize}
  \item The divergence of the idealized process from $\mu$ due to its finiteness.
  \item The divergence of the actual process from the idealized process due to indivisibility.
\end{itemize}
The idea to compare the actual process to an idealized one was suggested in
\citep{RSW98} and was combined with convergence results of Markov
chains. Though we were inspired by the basic setup from \citep{RSW98} and the
probabilistic analysis from \citep{FS09}, our setting differs in a crucial
point: when dealing with the case $0 < \alpha < 1/2$, we get a delicate mixture
of the deterministic and the random model. For example, the random variables in our analysis are not symmetric anymore which leads to additional technicalities
(cf.~\lemref{independent}).

\medskip

Formally, let $\xi^{(t)}$ be the load vector of the idealized process at time $t$, then by the triangle inequality ($\one$ is  the all-one vector)
\[
    \|\by^{(t)}-\mu \one\|_\infty \leq \|\by^{(t)}-\xi^{(t)}\|_\infty + \|\xi^{(t)}-\mu\one \|_\infty.
\]

\begin{pro}\label{pro:TriangleIneq}
Let $G$ be a balancing network with matchings $M^{(1)},\ldots,M^{(T)}$.  Then,
\begin{itemize}
  \item $\|\xi^{(t_2)}-\mu\emph{\one}\|_\infty \leq \Lambda_2$,
  \item $\whp$ over the $\alpha$\nobreakdash-perturbation operation, \\$\|\by^{(t_2)}-\xi^{(t_2)}\|_\infty \leq (t_2-t_1)
   	  + 3 \left(\tfrac{1}{2} - \alpha \right) t_1 +\Lambda_1.$
\end{itemize}
\end{pro}
\thmref{upper} then follows. The proof of the first part of the proposition consists of standard spectral arguments and is given in \secref{spectral} for completeness.  The proof the second part is more involved and is given in
\secref{TriangleIneqTwo}.

\subsection{Proof of \proref{TriangleIneq}: Bounding $\|\xi^{(t_2)}-\mu\one\|_\infty$}
\label{sec:spectral}
Letting $\xi^{(0)}$ be the initial load vector, it is easily seen that
\[
    \xi^{(t_2)} = \xi^{(0)} \, \bP^{(1)} \, \bP^{(2)} \cdots \bP^{(t_2)},
\]
where $\bP^{(i)}$ is the matrix
corresponding to matching $M^{(i)}$ (as defined in \secref{TheModel}).
For simplicity let us abbreviate
\[
    \bP^{[i,j]}:= \bP^{(i)}\,\bP^{(i+1)}\,\bP^{(i+2)}\cdots \bP^{(j)}.
\]
Since $\bP^{[1,t_2]}$ is real valued and symmetric (as each of the $\bP^{(i)}$'s is),
it has $n$ real-valued eigenvalues $\lambda_1 \geq \lambda_2 \geq \ldots \geq \lambda_n$
whose corresponding eigenvectors $v_1,\ldots,v_n$ form an orthogonal basis of $\mathbb{R}^n$.
Next we observe that 
\[
    \xi^{(t)} - \mu\,\one
    =
    \xi^{(0)}\,\bP^{[1,t_2]}-\mu\,\one \, \bP^{[1,t_2]}
    =
    (\xi^{(0)}-\mu\one)\,\bP^{[1,t_2]},
\]
since $\one$ is an eigenvector of $\bP^{[1,t_2]}$
corresponding to $\lambda_1=1$. Furthermore $\xi^{(0)}\cdot \one$ is just the total (initial) number of tokens,
and therefore by definition we get $(\xi^{(0)}-\mu\one)\cdot \one = 0$.
Finally, let us project $\xi^{(0)}-\mu\one$ onto $v_1,\ldots,v_n$, that is, write
$\xi^{(0)}-\mu\one= \sum_{i=2}^n c_i v_i$ ($c_1=0$ as we said). For our goal to bound $\|\xi^{(t_2)}-\mu\one \|_\infty$,
it suffices to bound $\|\xi^{(t_2)}-\mu\one \|$ (recall that $\|.\|$ refers
to the $\ell_2$-norm) as for every vector $z$, $\|z\|_\infty \leq \|z \|.$ By the above,
\begin{align*}
    \big\| \xi^{(t_2)}-\mu\one \big\|
    &= \big\|  (\xi^{(0)}-\mu\one) \bP^{[1,t_2]} \big\|\\
    &= \Bigg\| \sum_{i=2}^n c_i v_i \cdot \bP^{[1,t_2]} \Bigg\|\\
    &= \Bigg\| \sum_{i=2}^n c_i \lambda_i v_i \Bigg\|.
\end{align*}
Recall that $\lambda(\bP^{[1,t_2]})$ denotes the second largest eigenvalue of $\bP^{[1,t_2]}$ in absolute value.
By the definition of the $\ell_2$-norm, and using the fact that the $v_i's$ form an orthogonal basis, the latter equals
\begin{align*}
    \left(\sum_{i=2}^n c_i^2 \lambda_i^2 \|v_i\|^2 \right)^{1/2}
    \leq& \lambda(\bP^{[1,t_2]}) \cdot \left(\sum_{i=2}^n c_i^2 \|v_i\|^2 \right)^{1/2}\\
    =& \lambda(\bP^{[1,t_2]}) \cdot \|\xi^{(0)}-\mu\one \|\\
    \leq&  \lambda(\bP^{[1,t_2]})\, K \sqrt{n}.
\end{align*}

\subsection{Proof of \proref{TriangleIneq}: Bounding $\|\by^{(t_2)}-\xi^{(t_2)}\|_\infty$}
\label{sec:TriangleIneqTwo}

The proof of this part resembles in nature the proof of \thmref{lower}.
Assuming an ordering of $G$'s vertices, for a balancer $\bal$ in round $t$,
$\bal=(u,v)$, $u < v$, we set $\Phi^{(t)}_{u,v}=1$ if the initial  direction (before the perturbation) is $u \to v$ and $-1$ otherwise
(in the lower bound we considered the all-up orientation thus we had no use
of these variables). As in \secref{lower}, for a balancer $\bal=(u,v)$ in round $t$, the random variable $\Psi^{(t)}_{u,v}$ is $-1/2$ if the balancer
is perturbed and $1/2$ otherwise. Using these notations we define a \emph{rounding vector} $\rho^{(t)}$, which
accounts for the rounding errors in step $t$. Formally,
\[
    \rho_u^{(t)} =
    \begin{cases}
        \Odd(y_u^{(t-1)}+y_v^{(t-1)})\cdot\Psi^{(t)}_{u,v}\cdot\Phi^{(t)}_{u,v}
        \hspace{-3.5cm}
        \\
        &\text{if $u$ and $v$ are matched in $M^{(t)}$ and $u<v$},\\
        -\Odd(y_u^{(t-1)}+y_v^{(t-1)})\cdot\Psi^{(t)}_{v,u}\cdot\Phi^{(t)}_{v,u}
        \hspace{-3.5cm}
        \\        
        &\text{if $u$ and $v$ are matched in $M^{(t)}$ and $u>v$}, \\
        0 & \text{if $u$ is unmatched.}
    \end{cases}
\]

\noindent
Now we can write the actual process as follows:

\begin{equation}\label{eq:ActualProcess}
\by^{(t)}=\by^{(t-1)}\bP^{(t)}+\rho^{(t)}.
\end{equation}

\noindent
Let $\MtEven$ be the set of balancers at time $t$ with no excess token, and $\MtOdd$ the ones with. We can rewrite $\rho^{(t)}$ as follows:
\[
    \rho^{(t)}=\tsum_{(u,v)\in \MtOdd} \Psi^{(t)}_{u,v}\cdot \Phi^{(t)}_{u,v}\cdot\big( e_u - e_v \big).
\]

\noindent
Unfolding \eq{ActualProcess}, yields then
\begin{equation*}
    \by^{(t)} = \by^{(0)}\bP^{[1,t]} + \sum_{i=1}^t \rho^{(i)} \bP^{[i+1,t]},
\end{equation*}
where
$\bP^{[2,1]} = \mathbf{I}$.
\noindent
Observe that $\by^{(0)}\bP^{[1,t]}$ is just $\xi^{(t)}$ (as $\xi^{(0)}=\by^{(0)}$), and therefore
\begin{align*}
    \by^{(t)}-\xi^{(t)}
    &= \tsum_{i=1}^t \rho^{(i)} \bP^{[i+1,t]} \\
    &= \tsum_{i=1}^t
           \tsum_{(u,v) \in \MlOdd}
           \Psi^{(i)}_{u,v} \cdot \Phi^{(i)}_{u,v}\cdot(e_u-e_v)\cdot \bP^{[i+1,t]}.
\end{align*}
In turn,
\begin{equation}\label{eq:xminusxi}
    \big(\by^{(t)}-\xi^{(t)}\big)_w= \tsum_{i=1}^t
           \tsum_{(u,v)\in \MlOdd}
           \Psi^{(i)}_{u,v}\cdot \Phi^{(i)}_{u,v}\cdot\left(\bP^{[i+1,t]}_{u,w} - \bP^{[i+1,t]}_{v,w}\right).
\end{equation}

Our next task is to bound \eq{xminusxi} to receive the desired term from \proref{TriangleIneq}.
We do that similar in spirit to the way we went around in
\secref{ProofOfpro:logn_LowerBound}. We break this sum into its first $t_1$ summands
(whose expected sum we calculate and to which we apply a large-deviation-bound). The
remaining $(t-t_1)$ terms are bounded deterministically.

One major difficulty in the general setting is that \lemref{BalancersOddWithProbHalf} (which was crucial
in the proof of \thmref{lower}) does not hold in general as its
proof makes substantial use of the special structure of the $\ccc$. 

\medskip

\noindent Equation \ref{eq:xminusxi} with $t=t_2$ yields
    \begin{eqnarray*}
    (x^{(t_2)}-\xi^{(t_2)})_w
           &=&
           \tsum_{i=1}^{t_2}
           \tsum_{(u,v)\in \MlOdd}
           \Psi_{u,v}^{(i)} \, \Phi_{u,v}^{(i)}
           \,
           \Big(
           \bP^{[i+1,t_2]}_{u,w} -
           \bP^{[i+1,t_2]}_{v,w}
           \Big).
     \end{eqnarray*}
  With $e_{u,v}$ denoting the row-vector with $+1$ at $u$-th column
  and $-1$ at $v$-th column and zeros elsewhere, we can rewrite and split this equation as follows:
  \begin{eqnarray*}
     (x^{(t_2)}-\xi^{(t_2)})_w  &=&
           \tsum_{i=1}^{t_1}
           \tsum_{(u,v)\in \MlOdd}
           \Psi_{u,v}^{(i)} \, \Phi_{u,v}^{(i)}
           \,
           \big( e_{u,v} \bP^{[i+1,t_2]} e_{w} \big) \\
           & & +
           \tsum_{i=t_1+1}^{t_2}
           \tsum_{(u,v)\in \MlOdd}
           \Psi_{u,v}^{(i)} \, \Phi_{u,v}^{(i)}
           \,
           \big( e_{u,v} \bP^{[i+1,t_2]} e_{w} \big) .
    \end{eqnarray*}
    Clearly,
    \begin{eqnarray*}
     & & \tsum_{i=t_1+1}^{t_2}
           \tsum_{(u,v)\in \MlOdd}
           \Psi_{u,v}^{(i)} \, \Phi_{u,v}^{(i)}           \,
           \big( e_{u,v} \bP^{[i+1,t_2]} e_{w} \big) \\
	     &=&  \tsum_{i=t_1+1}^{t_2} \left(
           \tsum_{(u,v)\in \MlOdd}
            \Psi_{u,v}^{(i)} \, \Phi_{u,v}^{(i)}
           \,e_{u,v}  \right) \cdot \left(  \bP^{[i+1,t_2]}   \,  e_{w}  \right).
     \end{eqnarray*}
     Observe that $
           \tsum_{(u,v)\in \MlOdd}
            \Psi_{u,v}^{(i)} \, \Phi_{u,v}^{(i)}
           \,  e_{u,v} $ is a vector all of whose entries are bounded by $1$ in absolute value. Since $\bP^{[i+1,t_2]}$ is a stochastic matrix, the sum of all entries of the $w$-th column of $\bP^{[i+1,t_2]}$, which is $|\bP^{[i+1,t_2]}   \,  e_{w}|_1$, is exactly one, hence
	\begin{eqnarray*}
	 \, \left| \tsum_{i=t_1+1}^{t_2} \left(
           \tsum_{(u,v)\in \MlOdd}
            \Psi_{u,v}^{(i)} \, \Phi_{u,v}^{(i)}
           \,e_{u,v}  \right) \cdot \left(  \bP^{[i+1,t_2]}   \,  e_{w}  \right)  \right|
    &\leq&   ({t_2}-t_1).
    \end{eqnarray*}
    It remains
    to bound \[ W_{\Odd} :=
           \tsum_{i=1}^{t_1}
           \tsum_{(u,v)\in \MlOdd}
           \Psi_{u,v}^{(i)} \, \Phi_{u,v}^{(i)}
           \,
           \big( e_{u,v} \bP^{[i+1,t_2]} e_{w} \big). \]
	   Because $W_{\Odd}$ is not necessarily a sum of independent
	   random variables, it will be more convenient to work with the following quantity (which is a sum of independent random variables, as it assumes that every balancer gets an excess token),
    \begin{align*}
    W:=
           \tsum_{i=1}^{t_1}
           \tsum_{(u,v)\in \Ml}
           \Psi_{u,v}^{(i)} \, \Phi_{u,v}^{(i)}           \,
           \big( e_{u,v} \bP^{[i+1,t_2]} e_{w} \big).
    \end{align*}
     So our strategy is first to bound the deviation of $W$ from its mean by Hoeffdings Bound and then apply the following lemma (whose proof is in \secref{independent}), which justifies using $W$ instead of $W_{\Odd}$.

\begin{lem}
    \label{lem:independent}
    Fix $0 \leq \alpha \leq \tfrac{1}{2}$. For all $t_1,t_2$ with $t_1 < t_2\leq T$ and
    arbitrary weights $w_{u,v}^{(i)}\in\R$, let
    \begin{align*}
    W_{\Odd}&:=\tsum_{i=t_1}^{t_2}
           \tsum_{(u,v)\in \MlOdd} \Psi_{u,v}^{(i)} \, \Phi_{u,v}^{(i)} \, w_{u,v}^{(i)},\\
       W_{\Even}&:=\tsum_{i=t_1}^{t_2}
           \tsum_{(u,v)\in \MlEven}
           \Psi_{u,v}^{(i)} \, \Phi_{u,v}^{(i)} \, w_{u,v}^{(i)},\\
       W&:=\tsum_{i=t_1}^{t_2}
           \tsum_{(u,v)\in \Ml}
           \Psi_{u,v}^{(i)} \, \Phi_{u,v}^{(i)} \, w_{u,v}^{(i)}.
    \end{align*}
    Then for any $\delta > 0$,
    \begin{align*}
        \Pr{  \left|W_{\Odd}\right| \geq \delta + 2 \max\{ |\Ex{W^+}|, |\Ex{W^-}|  \} }
        \leq  16 \, \Pr{|W| \geq \delta},
    \end{align*}
    where
     \begin{align*}
   W^{+}&:=\tsum_{i=t_1}^{t_2}
       \tsum_{\substack{ (u,v)\in \Ml \\ \Phi_{u,v}^{(i)} w_{u,v}^{(i)} > 0   }}
       \Psi_{u,v}^{(i)} \,\Phi_{u,v}^{(i)}
       w_{u,v}^{(i)}, \\
     W^{-}&:=\tsum_{i=t_1}^{t_2}
       \tsum_{\substack{ (u,v)\in \Ml \\ \Phi_{u,v}^{(i)} w_{u,v}^{(i)} < 0   }}
       \Psi_{u,v}^{(i)} \,\Phi_{u,v}^{(i)}
       w_{u,v}^{(i)}.
\end{align*}
\end{lem}
In order to apply Lemma \ref{lem:independent} we first need to bound the expectation:
    \begin{align*}
     \left|\Ex{W} \right|&= \left|
        \tsum_{i=1}^{t_1}
        \tsum_{(u,v)\in \Ml}
         \Ex{ \Psi_{u,v}^{(i)} \, \Phi_{u,v}^{(i)} }
        \,
        e_{u,v} \bP^{[i+1,t_2]} e_{w}\right| \\
     &\leq
	\max_{i=1}^{t_1}
        \max_{(u,v)\in \Ml} \left| \Ex{ \Psi_{u,v}^{(i)} \, \Phi_{u,v}^{(i)}}  \right| \\
        &\phantom{\leq \max_{i=1}^{t_1} \max_{(u,v)\in \Ml}}
        \cdot
	    \left|\tsum_{i=1}^{t_1} \tsum_{(u,v)\in \Ml} e_{u,v} \bP^{[i+1,t_2]} e_{w} \right|.
      \end{align*}
      As we explained before, the last term is at most $t_1$, and for any $i \in [1,t_1]$ and $(u,v) \in \Ml$,
       $\Ex{ \Psi_{u,v}^{(i)} \, \Phi_{u,v}^{(i)} } \leq (1/2-  \alpha).$ Thus,
      \begin{align}
    \left| \Ex{W} \right| \leq
        t_1 \, ( 1/2 -\, \alpha ). \notag
        \label{eq:EW}
    \end{align}
    Similarly,
    $\left| \Ex{W^+} \right| \leq t_1 \, ( 1/2 -\, \alpha )$
    and $\left| \Ex{W^-} \right| \leq t_1 \, ( 1/2 -\, \alpha )$.

    To apply Hoeffdings bound on $W$, we bound the sum of the squared ranges of
    the involved random variables as follows:
    \begin{align*}
    & \tsum_{i=1}^{t_1}
           \tsum_{(u,v)\in \Ml}
           \Range{\Psi_{u,v}^{(i)} \, \Phi_{u,v}^{(i)}
           \, \cdot e_{u,v}\bP^{[i+1,t_2]}e_w
           }^2 \\
    &\leq \tsum_{i=1}^{t_1}
           \tsum_{(u,v)\in \Ml}
           \left(e_{u,v}\bP^{[i+1,t_2]}e_w \right)^2
     =:  \gamma.
     \end{align*}
    Now by Hoeffdings Bound,
    \begin{align*}
     & \Pr{|W| \geq |\Ex{W}| + \varepsilon}
        \leq \Pr{ |W - \Ex{W}| \geq \varepsilon}\\
        &\leq 2 \, \exp \left(- 2 \varepsilon^2 \Big/ \tsum_{i=1}^{t_1}
           \tsum_{(u,v)\in \Ml}
           \Range{\Psi_{u,v}^{(\ell)} \, \Phi_{u,v}^{(i)}
           \,
           \Big(
           e_{u,v}\bP^{[i+1,t_2]}e_w
           \Big)}^2  \right)\\
        &\leq 2 \, \exp \left(- 2 \varepsilon^2/\gamma \right).
    \end{align*}
    Choosing $\epsilon=\sqrt{\gamma \, \log n}$ we get
    $
        \Pr{ |W | \geq | \Ex{W}| + \sqrt{\gamma  \log n}} \leq 2 n^{-2}.
    $
    Hence by \lemref{independent} we obtain with
    $\delta=|\Ex{W}| + \sqrt{\gamma \, \log n}$ that
    \begin{align*}
        &\Pr{ |W_{\Odd}| \geq |\Ex{W}| + 2 \max \{ | \Ex{W^+}|, |\Ex{W^-}| \}   + \sqrt{\gamma \log n}}\\
        &\leq 16  \Pr{ |W| \geq |\Ex{W}| + \sqrt{\gamma \log n}} \\
        &\leq 32 \, n^{-2}.
    \end{align*}
    As $t_1 \, (1/2 - \alpha)$ is an upper bound on each of $ |\Ex{W}|,  | \Ex{W^+}|$ and $|\Ex{W^-}|$,
    this readily implies
    \begin{align*}
     \Pr{ |W_{\Odd}| \geq 3  \, (1/2- \alpha) t_1  + \sqrt{\gamma  \log n}} &\leq 32 n^{-2}.
    \end{align*}
    Finally, taking the union bound, we conclude by \eq{xminusxi} that for all vertices $w \in V$
    \begin{align*}
      \Pro{ \bigwedge_{w \in V}
            \left( \left| \xi_w^{(t_2)} - x_w^{(t_2)} \right|
            \leq 3 \, ( 1/2 - \alpha ) t_1
                 + \gamma\sqrt{  \log n} \right) }
            &\geq 1 - 32 n ^{-1}.
         \qedhere
    \end{align*}


\subsection{Proof of \lemref{independent}}
\label{sec:independent}

Before we begin the proof of \lemref{independent},
we require the following two technical lemmas.

\begin{lem}
    \label{lem:simplelemma}
    For $0\leq \alpha \leq 1/2$,
    arbitrary $w_i \in \mathbb{R}$,
    and independent random variables $X_i$ that are $-1$ with probability $\alpha$
    and $1$ otherwise, let $X:=\sum_{i=1}^n w_i X_i$.  Then,
    \begin{itemize}
        \item If all $w_i \geq 0$, then
        $\Pro{ X < 0} \leq 1/2$,\\and for any number $\delta \geq 0$, $\Pro{ X > \delta} \geq \Pro{ X < - \delta}$.
        \item If all $w_i \leq 0$, then
        $\Pro{ X > 0} \leq 1/2$,\\and for any number $\delta \geq 0$, $\Pro{ X < - \delta} \geq \Pro{ X > \delta}$.
    \end{itemize}
\end{lem}
\begin{proof}
   Note that it suffices to prove the statement with all $w_i \geq 0$, as the case with $w_i \leq 0$ follows from the first
    by considering $X':= \sum_{i=1}^n w_i' X_i$ with $w_i' := -w_i$.
    Let $Y:=\sum_{i=1}^n w_i Y_i$ where each $Y_i$ is an independent and uniform
    random variable taking values in $\{-1,+1\}$
    (corresponding to $X$ with $\alpha=1/2$).
    As $Y$ is a sum of symmetrical distributed random variables,
    we have $\Pro { Y < 0} \leq 1/2$. Since for any $\alpha \leq 1/2$,
    $X$ stochastically dominates $Y$, we obtain
    $
        \Pro{ X < 0} \leq \Pro{Y < 0} \leq 1/2,
    $
    and the first claim of the lemma follows.

    The second claim is proven similarly by observing that
    \[
     \Pro{ X > \delta} \geq \Pro{ Y >  \delta}
        = \Pro{ Y < -\delta}
        \geq \Pro{ X < - \delta}.
        \qedhere
    \]
\end{proof}

The following lemma bounds the probability that a sum $X:=\sum_{i=1}^n w_i X_i$ deviates
by more than a factor two from its expectation.

\begin{lem}
    \label{lem:posrattail}\label{lem:negrattail}
    For $0\leq \alpha \leq 1/2$,
    arbitrary $w_i\in\Q$,
    and $X_i$ independent random variables that are $-1$ with probability $\alpha$
    and $1$ otherwise, let $X:=\sum_{i=1}^n w_i X_i$.  Then,
    \begin{enumerate}
     \item If all $w_i \geq 0$, then  $
        \Pro{ X > 2 \Ex{X}} \leq 7/8.
    $
     \item If all $w_i \leq 0$, then  $
        \Pro{ X < 2 \Ex{X}} \leq 7/8.
    $
    \end{enumerate}

\end{lem}
\begin{proof}As before, it suffices to prove the first statement.
    Moreover, we may focus on the case $\Ex{X} \neq 0$, as otherwise $X$ is symmetrically
    distributed around $0$.
    Finally, we also assume the weights to be integral, that is, $w_i\in\Z_{\geq0}$.
    Rational weights $w_i\in\Q_{\geq0}$ can be easily reduced to integral
    weights by multiplying all weights with their least common multiple
    of the denominators and applying the bound for the integral case.

    For the sake of contradiction, suppose that $\Pr{ X > 2 \Ex{X}} > 7/8$.  As $0\leq \alpha \leq 1/2$, we have
    $\Ex{X}= \sum_{i=1}^n w_i (1 - 2 \alpha) \geq0$ and
    hence $\Pr{ X \leq 0} \leq \Pr{ X \leq 2 \Ex{X}} < 1/8$.
    Let $k>2\Ex{X}$ be such that
    $\Pr{ X \geq k} \geq 1/8,$ and $\Pr{ X  >   k}    < 1/8.$

    As we assumed the weights $w_i$ to be integral, we can use the following two
    well-known counting tricks:
    \begin{align*}
        \sum_{x\geq k} x \Pr{X=x}
            &= \sum_{x\geq k} \sum_{y=1}^x \Pr{X=x}\\
            &= \sum_{y\geq 1} \sum_{x\geq\max(k,y)} \Pr{X=x}\\
            &= \sum_{x\geq 1} \Pr{X\geq\max(k,x)},
	    \end{align*}
  and similarly, $ \sum_{x\leq -k} -x \Pr{X=x} = \sum_{x\geq 1} \Pr{X\leq -\max(k,x)}$.

    We use both to obtain
    \begin{align*}
        \sum_{x\geq k} x \Pr{X=x}
            &= \sum_{x\geq 1} \Pr{X\geq\max(k,x)} \\
            &= \sum_{1\leq x<k} \Pr{X\geq k}
               + \sum_{x\geq k} \Pr{X\geq x} \\
            &> \sum_{1\leq x<k} \left( \Pr{X\leq 0} + \Pr{X\geq k}-1/8 \right)\\
            &\phantom{>} + \sum_{x\geq k} \Pr{X\geq x}. \\
    \end{align*}
    By \lemref{simplelemma} we now get
    \begin{align*}
        \sum_{x\geq k} x \Pr{X=x}
            &\geq k\,(\Pr{X\geq k}-1/8)
               + \sum_{1\leq x<k} \Pr{X\leq -x}\\
            &\phantom{\geq} + \sum_{x\geq k} \Pr{X\leq -x} \\
            &= k\,( \Pr{X\geq k}-1/8)
               + \sum_{x\geq 1} \Pr{X\leq -x} \\
            &= k\,(\Pr{X\geq k}-1/8)
               + \sum_{x\leq -1} -x \Pr{X=x}
    \end{align*}

    \noindent
    Plugging this in the definition of $\Ex{X}$, we get
    \begin{align*}
        \Ex{X} &= \sum_{x} x \Pr{X=x}\\
            &= \sum_{x\leq -1} x \Pr{X=x} +
              \sum_{0\leq x<k} x \Pr{X=x} +
              \sum_{x\geq k} x \Pr{X=x} \\
            &\geq k\,(\Pr{X\geq k}-1/8) + \sum_{0\leq x<k} x \Pr{X=x}.
    \end{align*}
    Using above assumptions on $k$ and $\Pr{X \geq 2\Ex{X}} > 7/8$ we now arrive at the desired contradiction,
    \begin{align*}
        \Ex{X}
            &> 2\Ex{X}\,(\Pr{X\geq k}-1/8) + \sum_{2\Ex{X}\leq x<k} x \Pr{X=x} \\
            &\geq 2\Ex{X}\,\bigl( \Pr{X\geq k}-1/8 + \Pr{2\Ex{X}\leq X<k}\bigr) \\
            &\geq 2\Ex{X}\,\left( \Pr{X\geq2\Ex{X}}-1/8  \right) \\
            &\geq 3/2 \Ex{X}.
            \qedhere
    \end{align*}
\end{proof}

\noindent
We are now ready to prove \lemref{independent}.
In the following, we will use subsums of $W$ denoted as $W_{\Even}^+, W_{\Even}^-,
W_{\Odd}^+,W_{\Odd}^-$ which are defined by combining the previous definitions
in a natural way, e.~g.,
\[
 W_{\Odd}^+ := \tsum_{i=t_1}^{t_2}
       \tsum_{\substack{ (u,v)\in \MlOdd \\ \Phi_{u,v}^{(i)} w_{u,v}^{(i)}  > 0   }}
       \Psi_{u,v}^{(i)} \,\Phi_{u,v}^{(i)}
       w_{u,v}^{(i)}.
\]

\noindent
Let $\Mtt=\bigcup_{i=t_1}^{t_2} \Ml$ be the set of all matching edges in the given time span.
Let $\MttOdd$ be the set of all odd ones.
Moreover, let us simply write $\cM$ for a particular
assignment of $\Odd(.)$ for all balancers in $\Mtt$.
We now begin by bounding $W_{\Odd}$ in terms of $W$.
By the definition of conditional probability and the law of total probabilities,
it follows for arbitrary $\delta_1,\delta_2\in\R$,
\begin{align*}
    &\Pr{W<\delta_2 \mid W_{\Odd}\geq \delta_1}
    =
    \frac
        {\Pr{W<\delta_2 \wedge W_{\Odd}\geq\delta_1}}
        {\Pr{W_{\Odd}\geq\delta_1}}\notag\\
    &\leq
    \frac
        {\Pr{W-W_{\Odd}<\delta_2-\delta_1 \wedge W_{\Odd}\geq\delta_1}}
        {\Pr{W_{\Odd}\geq\delta_1}}\notag\\
    &=
    \sum_{\myM}
        \frac{
            \Pr{\MttOdd=\myM}
            \Pr{W_{\Even}<\delta_2-\delta_1 \wedge
            W_{\Odd}\geq\delta_1 \mid \MttOdd=\myM}
        }{\Pr{W_{\Odd}\geq\delta_1}} \notag.
\end{align*}

\noindent
Note that for fixed $\myM$, the ranges of $W_{\Odd}$ and $W_{\Even}$ are determined and
therefore the probability spaces are disjoint.  This implies that $W_{\Odd}$ and $W_{\Even}$
are independent conditioned on $\MttOdd=\myM$.  Using this observation, we get
\begin{align}
    &\Pr{W<\delta_2 \mid W_{\Odd}\geq \delta_1} \notag\\[-.3cm]
    &\qquad\leq
    \sum_\myM
        \Pr{\MttOdd=\myM}
        \overbrace{\Pr{W_{\Even}<\delta_2-\delta_1 \mid \MttOdd=\myM}}^{(*)}
        \label{eq:ThomasFavorite4}\\[-.1cm]
        &\qquad \qquad \qquad \cdot\Pr{W_{\Odd}\geq\delta_1 \mid \MttOdd=\myM}
         \Big/ \Pr{W_{\Odd}\geq\delta_1}.\notag
\end{align}

\noindent
By plugging
$\delta_2 := \delta + 2 \, \Ex{W^{-}}$
and
$\delta_1 := \delta$
in (*) we obtain
\begin{align}
    &\Pr{W_{\Even} <2 \, \Ex{W^{-}} \mid \MttOdd=\myM} \notag\\
	&\qquad= 1 - \Pr{W_{\Even} \geq 2 \, \Ex{W^{-}} \mid \MttOdd=\myM} \notag \\
	&\qquad\leq 1 - \Pr{ W^+_{\Even} \geq 0 \mid \MttOdd=\myM} \notag\\
	&\qquad\phantom{\leq}\cdot \Pr{ W^-_{\Even} \geq 2 \, \Ex{W^{-}} \mid \MttOdd=\myM}. \notag
\end{align}

\noindent
Observing that $\Ex{\Psi_{u,v}} = \Ex{\Psi_{u,v} \mid \MttOdd=\myM}$
for $(u,v)\in \MlEven$, we get $\Ex{W^{-}} \leq \Ex{W^{-}_{\Even} \mid  \MttOdd=\myM}$
and hence by \lemrefs{simplelemma}{negrattail},
\begin{align}
	& \Pr{W_{\Even} <2 \, \Ex{W^{-}} \mid \MttOdd=\myM} \notag \\
	&\qquad\leq 1 - \Pr{ W^+_{\Even} \geq 0 \mid \MttOdd=\myM} \cdot \notag\\
	&\qquad\qquad\Pr{ W^-_{\Even} \geq 2 \, \Ex{W^{-}_{\Even} \, \mid \, \MttOdd=\myM} \mid \MttOdd=\myM} \notag \\
	&\qquad\leq 1 - \tfrac{1}{2} \cdot \tfrac{1}{8}
    = \tfrac{15}{16}.
    \label{eq:applytail}
\end{align}
Plugging this into \eq{ThomasFavorite4} yields
\begin{align*}
    & \Pr{W< \delta + 2 \Ex{W^-}  \mid W_{\Odd}\geq \delta} \\
    &\qquad\leq \frac{15}{16} \, \sum_\myM
        \frac{
            \Pr{\MttOdd=\myM}
                \Pr{W_{\Odd}\geq\delta_1 \mid \MttOdd=\myM}
        }{\Pr{W_{\Odd}\geq\delta_1}} = \frac{15}{16},
\end{align*}
and hence
\begin{align}
   &\Pr{ W \geq \delta + 2 \Ex{W^-}}\notag\\
   &\qquad\geq \Pr{ W \geq \delta + 2 \Ex{W^-} \wedge W_{\Odd} \geq \delta} \notag\\
   &\qquad= \Pr{W_{\Odd} \geq \delta} \cdot \Pr{W \geq \delta + 2 \Ex{W^-}
      \, \mid \, W_{\Odd} \geq \delta} \notag\\
   &\qquad\geq \tfrac{1}{16} \cdot \Pr{ W_{\Odd} \geq \delta}. \label{eq:firststep}
\end{align}

\noindent
It remains to lower bound the deviation of $W_{\Odd}$ in terms of $W$ in a similar fashion.
As before, we derive
\begin{align}
    &\Pr{W>\delta_2 \mid W_{\Odd}\leq \delta_1} \notag\\[-.3cm]
    &\qquad\leq
    \sum_\myM
        \Pr{\MttOdd=\myM}
        \overbrace{\Pr{W_{\Even}>\delta_2-\delta_1 \mid \MttOdd=\myM}}^{(*)}
        \label{eq:ThomasFavorite2}\\[-.1cm]
        &\qquad \qquad \qquad \cdot\Pr{W_{\Odd}\leq\delta_1 \mid \MttOdd=\myM}
         \Big/ \Pr{W_{\Odd}\leq\delta_1}.\notag
\end{align}
We now choose
$
    \delta_2 := -\delta + 2 \, \Ex{W^{+}},$ and
$\delta_1 := -\delta$.
As before in \eq{applytail}, we can now use \lemrefs{simplelemma}{posrattail}
to get for (*) that
\begin{align*}
    \Pr{W_{\Even} > 2 \, \Ex{W^{+}} \mid \MttOdd=\myM}
    &\leq \tfrac{15}{16}.
\end{align*}
 Plugging this into \eq{ThomasFavorite2} yields
\[
    \Pr{W> - \delta + 2 \Ex{W^+} \mid W_{\Odd}\leq -\delta}
    \leq
    15/16
\]
 and hence
\begin{align}
    &\Pr{ W \leq -\delta + 2 \Ex{W^+}}\notag\\
    &\qquad\geq \Pr{ W \leq -\delta + 2 \Ex{W^+} \wedge W_{\Odd} \leq -\delta} \notag \\
   &\qquad= \Pr{W_{\Odd} \leq -\delta} \cdot \Pr{W \leq -\delta + 2 \Ex{W^+}   \, \mid \, W_{\Odd} \leq -\delta} \notag \\
   &\qquad\geq \tfrac{1}{16} \cdot \Pr{ W_{\Odd} \leq -\delta}  \label{eq:secondstep}.
\end{align}
Combining \eq{firststep} and \eq{secondstep}, we have shown for any $\delta \geq 0$,
\begin{align*}
 &\Pro{ |W_{\Odd}| \geq \delta} \\
 &\qquad= \Pro{ W_{\Odd} \geq \delta} + \Pro{ W_{\Odd} \leq -\delta} \\
 &\qquad\leq 16 \, \Pro{ W \geq \delta + 2 \Ex{W^{-}}} + 16 \, \Pro{ W \leq - \delta + 2 \Ex{W^{+}}} \\
 &\qquad\leq 16 \, \Pro{ |W| \geq \delta - 2 \max \{ |\Ex{W^+}|, |\Ex{W^-}| \} }.
\end{align*}
Adding $2 \max \{ |\Ex{W^+}|, |\Ex{W^-}| \}$ to $\delta$ gives the assertion of the lemma.


\section{Deriving the Upper Bound for the Periodic Case}
\label{sec:DerivationOfPeriodicCase}

Consider the network after $\frac{2 \log(Kn)}{1-\lambda(\bQ)}$ (periodic) repetitions of the balancing network (each repetition consists of $d$ rounds, so the network consists of a total of $T=d\cdot\frac{2 \log(Kn)}{1-\lambda(\bQ)}$
matchings). $\bQ$ is the so-called round matrix corresponding to the application of $d$ consecutive matchings
(recall that we use the following abbreviations $\bP^{[t_1,t_2]}=\tprod_{i=t_1}^{t_2} \bP^{(i)}$).
This is indeed a special case of our general scheme where the matrices $\bP^{(1)},\bP^{(2)},\ldots,\bP^{(T)}$ are applied,
but now $\bQ=\bP^{[1,d]}=\bP^{[d+1,2d]}$ and so on.
Recall the notation that $e_{u,v}$ denotes the row-vector with $+1$ at $u$-th column and $-1$ at $v$-th column.
\medskip
In \thmref{upper} we plug $t_2=T=\frac{2d\log(Kn)}{1-\lambda(\bQ)}$
and $t_1=T-\frac{2d \log\log (n)}{1-\lambda(\bQ)}$.
Hence,
   \begin{align*}
    \Lambda_1 &= \max_{w\in V} \sqrt{\log n \, \tsum_{i=1}^{t_1}
           \tsum_{(u,v)\in \Ml}
           \left(e_{u,v}\bP^{[i+1,t_2]}e_w \right)^2}.
   \end{align*}
  Let $v_1^{(i)},v_2^{(i)},\ldots,v_n^{(i)}$ be the eigenvectors of $\bP^{[i+1,t_2]}$ forming an orthogonal basis of $\mathbb{R}^n$ and let $\lambda_1^{(i)},\lambda_2^{(i)},\ldots,\lambda_n^{(i)}$ be the corresponding eigenvalues.
  Since $\bP^{[i+1,t_2]}$ is a stochastic matrix, $v_1^{(i)} = \one$.
  Write $e_w = \sum_{j=1}^n c_j^{(i)} v_j^{(i)}$. Using this, we rewrite $\Lambda_1$ as follows (for short we drop
  the ``$\max_w$'' part as our final result will not depend on $w$),
  \begin{align*}
   \Lambda_1 &= \sqrt{\log n \, \tsum_{i=1}^{t_1}
           \tsum_{(u,v)\in \Ml}
           \left(e_{u,v}\bP^{[i+1,t_2]} \sum_{j=1}^n c_j^{(i)} v_j^{(i)} \right)^2} \\
	   &= \sqrt{\log n \, \tsum_{i=1}^{t_1}
           \tsum_{(u,v)\in \Ml}
           \left(e_{u,v} \sum_{j=1}^n c_j^{(i)} \lambda_j^{(i)} v_j^{(i)}  \right)^2}.
  \end{align*}
Observe that $e_{u,v}$ is orthogonal to $v_{1}^{(i)}$. Therefore $c_1^{(i)}=0$ for every $i$ which gives
    \begin{align*}
   \Lambda_1 &= \sqrt{\log n \, \tsum_{i=1}^{t_1}
           \tsum_{(u,v)\in \Ml}
           \left(e_{u,v} \sum_{j=2}^n c_j^{(i)} \lambda_j^{(i)} v_j^{(i)}  \right)^2}.
   \end{align*}
  Define a vector $z^{(i)} := \sum_{j=2}^n c_j^{(i)} \lambda_j^{(i)} v_j^{(i)}$. The latter is then
     \begin{align*}
   \Lambda_1 &= \sqrt{\log n \, \tsum_{i=1}^{t_1}
           \tsum_{(u,v)\in \Ml}
           \left(e_{u,v} z^{(i)} \right)^2} \\
	   &= \sqrt{\log n \, \tsum_{i=1}^{t_1}
           \tsum_{(u,v)\in \Ml}
           \left(z_u^{(i)} - z_v^{(i)} \right)^2} \\
	   &\leq \sqrt{\log n \, \tsum_{i=1}^{t_1}
           \tsum_{(u,v)\in \Ml} 2 (z_u^{(i)})^2 + 2 (z_v^{(i)})^2} .
   \end{align*}
   Since $\Ml$ is a matching, for each $i$ each vertex is counted only
   at most once, thus
   \begin{equation}\label{eq:letzter}
    \Lambda_1 \leq 2 \sqrt{\log n \, \tsum_{i=1}^{t_1} \| z^{(i)} \|^2}.
   \end{equation}

   \noindent
   By standard calculation (cf.~\secref{spectral}), we have
   \begin{align*}
       \| z^{(i)} \|^2 &\leq \lambda(\bP^{[i+1,t_2]})^2 \cdot \| e_w \|^2 = \lambda(\bP^{[i+1,t_2]})^2.
   \end{align*}
   Plugging this into (\ref{eq:letzter}) and using the facts that for any two stochastic matrices $\bA$ and $\bB$, $\lambda(\bA \bB) \leq \lambda(\bA)$ and for any integer $k$, $\lambda(\bA^{k})= \lambda(\bA)^k$,
   \begin{align*}
    \Lambda_1 
    &\leq 2 \sqrt{ \log n \, \tsum_{i=1}^{t_1} \lambda( \bP^{[i+1,t_2]})^2} \\
    &\leq 2 \sqrt{ \log n \, \tsum_{i=1}^{t_1} \lambda( \bQ^{\lfloor (t_2-(i+1))/d \rfloor} )^2}.
   \end{align*}
  Regrouping the matrices according to the periods the latter reduces to
  \[ 
  \Lambda_1 
  \leq \sqrt{\log n \, \tsum_{i=1}^{t_1/d} d \, \lambda( \bQ )^{2 (t_2/d - (i+1)) }}.   \]
   With $\rho=((t_2- t_1)/d)-1$, we can upper bound $\Lambda_1$ by
     \[ 2 \sqrt{\log n \, d \, \tsum_{i=\rho}^{\infty}  \lambda( \bQ )^{2 i } } \leq 2  \,  \sqrt{\log n \, d \, \frac{\lambda(\bQ)^{2 \rho }}{1-\lambda(\bQ)^2}}.
     \]
     Plugging in our choices for $t_1$ and $t_2$ we end up with
     \begin{align*}
     \Lambda_1
     &\leq 2 \,  \sqrt{d \, \log n \, \frac{\lambda(\bQ)^{-1+4\log \log n/(1-\lambda(\bQ))}}{1-\lambda(\bQ)^2}} \\
     &\leq 2 \, \sqrt{d \, \log n \, \frac{ \exp( -3\log\log n)       }{1-\lambda(\bQ)^2}} \\ & \leq  2 \, \frac{1}{\log n} \, \sqrt{d \, \frac{ 1 }{1-\lambda(\bQ)^2}}\\
     &\leq \frac{2 \sqrt{d} }{\log n}\cdot \frac{ 1 }{1-\lambda(\bQ)}.
   \end{align*}

    \noindent
    With the same arguments, we get \begin{align*} \Lambda_2 &\leq \lambda( \bP^{[1,t_2]}) \, \sqrt{n} K \leq \, \sqrt{n} K/ (Kn)^2.                 \end{align*}
    Plugging all this into \thmref{upper} shows that at step $T= \tfrac{2 \log (Kn)}{1-\lambda_2(\bP)}$ the discrepancy is at most
    \begin{align*}
       & \frac{2d \, \log \log n}{1-\lambda(\bQ)} +
       \frac{3(1/2 - \alpha) \cdot 2d \, \log(Kn)}{1-\lambda(\bQ)} +  \frac{2 \sqrt{d} }{\log n}\cdot \frac{ 1 }{1-\lambda(\bQ)} + \frac{\sqrt{n} K}{(Kn)^2} \\ &= \Oh \left( \frac{d \log (Kn) }{1 - \lambda(\bQ) } \cdot \left(\frac{1}{2} - \alpha \right) + \frac{d \log \log n}{1 - \lambda(\bQ)}  \right).
    \end{align*}

\section{Deriving the Upper Bound for $\ccc_n$}\label{sec:DerivationOfCCC}

In \thmref{upper} we choose $t_2=\log n$ and $t_1=\log n - \log \log n$. For matrix multiplication
we use the following abbreviated form: $\bP^{[i,j]}:= \bP^{(i)}\bP^{(i+1)}\cdots \bP^{(j)}$.
First we observe that
\begin{itemize}
  \item  $\bP^{[k,t_2]}_{u,v}=2^k/n$ if wires $u$ and $v$ are at distance $\leq \log n-k$ and differ only in the last $\log n-k$ bits (in their
binary representation),
  \item otherwise, $\bP^{[k,t_2]}_{u,v}= 0$.
\end{itemize}
In particular this shows that $\bP^{[1,t_2]}$ is equal to the
all-$\big(\tfrac1n\big)$ matrix and thus
$\lambda \left(\bP^{[1,t_2]} \right)=0$ (that is, $\Lambda_2=0$).
Moreover, for any fixed wire $w$,
\begin{align*}
       \Lambda_1 & =  \sqrt{ \log n \,\tsum_{i=1}^{\log n - \log \log n} \tsum_{(u,v) \in \Ml} \left(
         e_{u,v} \left( \bP^{[i+1,t_2]} \right) e_w
        \right)^2}
        \\
        & = 
        \sqrt{ \log n\,\tsum_{i=1}^{\log n - \log \log n}  \tsum_{(u,v) \in \Ml}
        \left( \bP^{[i+1,t_2]}_{u,w} - \bP^{[i+1,t_2]}_{v,w}
        \right)^2}
         \\
        &\leq 
        \sqrt{ \log n \, \tsum_{i=1}^{\log n - \log \log n} \frac{n}{2^{i}} \left( \frac{2^{i+1}}{n} \right)^2} \\
        &= 4 \, \sqrt{ \log n \, \tsum_{i=1}^{\log n - \log \log n} \frac{2^i}{n}} \\
        &\leq 4 \, \sqrt{ \log n\cdot \frac{2}{\log n}} \\
        &= \Oh(1).
\end{align*}

\begin{figure}[tb]
    \centering
    \includegraphics[bb=62pt 259pt 550pt 542pt,width=0.66\textwidth,clip]{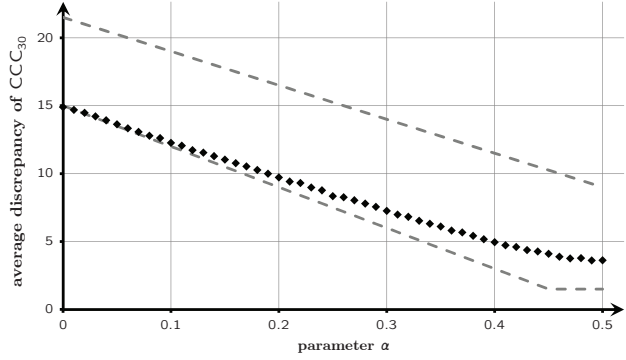}
    \caption{Discrepancy for various $\alpha$-values of $\ccc_{2^{30}}$ with random uniformly distributed input
             from $\{0,1,2,\ldots,2^{30}-1\}$.
             $\alpha=0$ corresponds to the adversarial model while
             $\alpha=1/2$ is the completely random model.
             The dotted line describes the experimental results,
             the broken lines are our theoretical lower and upper bounds.}
    \label{fig:experiments}
\end{figure}

\section{Experimental Result}
\label{sec:conclusions}

We examined experimentally how well a \ccc\ balances a random input. We 
implemented a $\ccc_{2^{30}}$ consisting of roughly one billion wires and thirty 
billion balancers.  The input was chosen independently uniformly at random from 
$\{0,1,2,\ldots,2^{30}-1\}$. For initial directions of the balancers all up and different values $\alpha$ between $0$ and $1/2$ we measured the resulted discrepancy.

\figref{experiments} presents the average over 100 runs, together
with theoretical lower and upper bounds.  As the input is random,
the so-far presented bounds can be tightened.  Following
the same lines of proof, one can easily show the
following slightly better bounds on the \emph{expected} discrepancy $\Delta$
in the random-input case:
\begin{itemize}
\item  $\Delta \leq (\frac{1}{2} - \alpha) \cdot (\log n - \lceil \log \log n \rceil) + \lceil \log \log n \rceil + 4,$
\item $\Delta \geq \max\{ (1/2 - \alpha) \log n,\;
      1/2\,(1-\tfrac1n)\,(\lfloor \log \log n \rfloor -1) \}$.
\end{itemize}
As these bounds are only used for visualization in \figref{experiments}
and the proofs are very similar to the above, they are omitted.


\newcommand{\FOCS}[2]{#1 Annual IEEE Symposium on Foundations of Computer Science (FOCS'#2)}
\newcommand{\STOC}[2]{#1 Annual ACM Symposium on Theory of Computing (STOC'#2)}
\newcommand{\SODA}[2]{#1 Annual ACM-SIAM Symposium on Discrete Algorithms (SODA'#2)}
\newcommand{\PODC}[2]{#1 Annual ACM Principles of Distributed Computing (PODC'#2)}
\newcommand{\ICALP}[2]{#1 International Colloquium on Automata, Languages, and Programming (ICALP'#2)}
\newcommand{\STACS}[2]{#1 International Symposium on Theoretical Aspects of Computer Science (STACS'#2)}
\newcommand{\SPAA}[2]{#1 ACM Symposium on Parallel Algorithms and Architectures (SPAA'#2)}
\newcommand{\MFCS}[2]{#1 International Symposium on Mathematical Foundations of Computer Science (MFCS'#2)}
\newcommand{\ISAAC}[2]{#1 International Symposium on Algorithms and Computation (ISAAC'#2)}
\newcommand{\WG}[2]{#1 Workshop of Graph-Theoretic Concepts in Computer Science (WG'#2)}
\newcommand{\SIROCCO}[2]{#1 International Colloquium on Structural Information and Communication Complexity (SIROCCO'#2)}
\newcommand{\IPDPS}[2]{#1 International Parallel and Distributed Processing Symposium (IPDPS'#2)}
\newcommand{\DISC}[2]{#1 International Symposium on Distributed Computing (DISC'#2)}
\newcommand{\RANDOM}[2]{#1 International Workshop on Randomization and Computation (RANDOM'#2)}
\newcommand{\IPPS}[2]{#1 IEEE International Parallel Processing Symposium (IPPS'#2)}

\end{document}